\documentclass[a4paper,fleqn]{cas-sc}
\usepackage{geometry}

\usepackage{times}
\usepackage{soul}
\usepackage{url}
\usepackage[T1]{fontenc}
\usepackage[utf8]{inputenc}

\usepackage{natbib}

\usepackage[small]{caption}
\usepackage{graphicx}
\usepackage{amsmath}
\usepackage{amsthm}
\usepackage{booktabs}
\usepackage{algorithm}
\usepackage[noend]{algorithmic}
\usepackage[switch]{lineno}

\usepackage{etoolbox}
\makeatletter
\patchcmd{\ps@pprintTitle}
  {Elsevier}
  {Journal of Computer and System Sciences}
  {}{}
\makeatother

\linenumbers

\usepackage{amsfonts}
\usepackage{amsmath}
\usepackage{amssymb}
\usepackage{amsthm}
\usepackage{colortbl}
\usepackage{bm}

\usepackage{enumitem}

\usepackage{xspace}
\usepackage{fancyhdr}
\usepackage{tcolorbox}

\usepackage{todonotes}
\usepackage{tcolorbox}

 \usepackage{newfloat}
\newif\iflong
\newif\ifshort

\longtrue

\iflong
\else
\shorttrue
\fi

\newtheorem{theorem}{Theorem}
\newtheorem{lemma}[theorem]{Lemma}
\newtheorem{corollary}[theorem]{Corollary}

\newtheorem{fact}[theorem]{Fact}

\newtheorem{definition}[theorem]{Definition}

\newtheorem{observation}{Observation}

\newcounter{myalgo}
\newcommand{\myalgo}[1]{\refstepcounter{myalgo}\label{#1}}

\newcommand{\NP}{\ensuremath{\mathtt{NP}}\xspace}

\newcommand{\Wtwo}{\ensuremath{\mathtt{W}[2]}\xspace}
\newcommand{\PP}{\ensuremath{\mathtt{P}}\xspace}

\newcommand{\FPT}{\ensuremath{\mathtt{FPT}}\xspace}
\newcommand{\XP}{\ensuremath{\mathtt{XP}}\xspace}
\newcommand{\APX}{\ensuremath{\mathtt{APX}}\xspace}
\newcommand{\W}[1][1]{\textsf{W[#1]}\xspace}
\newcommand{\Wh}[1][1]{\W[#1]-hard\xspace}

\newcommand{\bigoh}{\mathcal{O}}

\newcommand{\naesat}{\textsc{NAE3SAT}\xspace}
\newcommand{\monlinnaesat}{\textsc{MonLinNAE3SAT}\xspace}

\newcommand{\gcal}{\ensuremath{\mathcal{G}}\xspace}
\newcommand{\ecal}{\ensuremath{\mathcal{E}}\xspace}

\newcommand{\tuple}[1]{\ensuremath{\langle {#1} \rangle}\xspace}

\newcommand{\tmax}{\ensuremath{t_{\max}}\xspace}

\newcommand{\reach}{\ensuremath{\mathtt{reach}}\xspace}
\newcommand{\reachtime}{\ensuremath{\mathtt{reachtime}}\xspace}
\newcommand{\delay}{\ensuremath{\mathtt{shift}}\xspace}

\newcommand{\reachfast}{\textsc{ReachFast}\xspace}
\newcommand{\reachfasttotal}{\textsc{ReachFastTotal}\xspace}

\newcommand{\depth}{\mathtt{d}\xspace}
\newcommand{\tw}{\ensuremath{\operatorname{tw}}}
\newcommand{\tr}{\ensuremath{\operatorname{tr}}}

\newcommand{\hittngset}{\textsc{HittingSet}\xspace}

\def\tsc#1{\csdef{#1}{\textsc{\lowercase{#1}}\xspace}}
\tsc{WGM}
\tsc{QE}

\begin{document}
\let\WriteBookmarks\relax
\def\floatpagepagefraction{1}
\def\textpagefraction{.001}
\nolinenumbers
\shorttitle{}    

\shortauthors{}

\title{Minimizing Reachability Times on Temporal Graphs via Shifting Labels}

\tnotetext[t1]{A preliminary version of the results in this paper has appeared in IJCAI 2023~\cite{DBLP:conf/ijcai/DeligkasES23}. Argyrios Deligkas is supported by Engineering and Physical Sciences Research Council (EPSRC) grant EP/X039862/1.}

\author[1]{Argyrios Deligkas}
\ead{argyrios.deligkas@rhul.ac.uk}
\author[1]{Eduard Eiben}
\ead{eduard.eiben@rhul.ac.uk}
\author[2]{George Skretas}
\ead{georgios.skretas@hpi.de}
\address[1]{Royal Holloway, University of London, UK}
\address[2]{Hasso Plattner Institute, University of Potsdam, Germany}

\begin{abstract}
We study how we can accelerate the spreading of information in temporal graphs via shifting operations; a problem that captures real-world applications varying from information flows to distribution schedules.
In a temporal graph there is a set of fixed vertices and the available connections between them change over time in a predefined manner. 
We observe that, in some cases, shifting some connections, i.e., advancing or delaying them, can decrease the travel time from some vertex (source) to another vertex.
We study how we can minimize the maximum time a set of sources needs to reach every vertex, when we are allowed to shift some of the connections. 
If we restrict the allowed number of changes, we prove that, already for a single source, the problem is \NP-hard, and \Wtwo-hard when parameterized by the number of changes.
Then we focus on unconstrained number of changes. We derive a polynomial-time algorithm when there is a single source. 
When there are two sources, we show that the problem becomes \NP-hard; on the other hand, we design an \FPT algorithm parameterized by the treewidth of the graph plus the lifetime of the optimal solution, that works for any number of sources. 
Finally, we provide polynomial-time algorithms for several graph classes.
\end{abstract}


\maketitle

\section{Introduction}
\label{sec:intro}

Every day million pieces of information need to be disseminated and most often it is desirable to minimize their delivery time. 
In many cases, the diffusion of information depends on schedules of physical and online meetings between entities that form a dynamic network that changes over time, depending on their availability. 
Usually, these schedules are constrained due to physical limitations, laws, available infrastructure, and costs of extra resources and meetings. 
These constraints are usually unavoidable since it is difficult, if not impossible, to bypass them. 
On the other hand, careful changes on the existing scheduling timetable can significantly reduce the time a piece of information needs to reach every recipient. 

For an example motivated by real life, consider the scenario where there are three employees $A$, $B$, and $C$ in a university, where each one of them has some information that needs to reach every other person.
The central timetabling team has already arranged some meetings for them.
Between employees $A$ and $B$ there is one meeting at 9am and one at 11am. 
Between $B$ and $C$ there exists a meeting arranged at 8am and another one at 4pm. 
Then observe that under the current timetable, the information $A$ has will reach $C$ at 4pm -- $B$ will receive the information from $A$ at 9am and they will transmit it to $C$ at 4pm -- while the information $C$ has can reach $A$ at 9am. However, if we {\em delay} the 8am meeting between $B$ and $C$ to 10am, then the information of $A$ can reach $C$ at 10am, and the information of $C$ can still reach $A$ at 11am.
Figure~\ref{fig:example} depicts this example.

\begin{figure}[pos=h]
\includegraphics[scale=0.75]{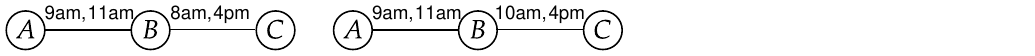}
\hspace{-1.2cm}
\caption{Left: The original timetable. Right: The modified timetable with the delayed meeting time.}
\label{fig:example}
\end{figure}

The inherent temporal nature of timetabling, combined with the existence of the underlying temporal network, allows us to concisely formulate the scheduling scenario described above as an optimization problem over a {\em temporal graph}~\cite{kempe02}. At a high level, a temporal graph consists of a fixed set of vertices and a timetable that defines available connections, or {\em temporal edges}, at any point in time. 
A temporal path 
is a sequence of temporal edges which, additionally to the usual connectivity, respects the time constraints, i.e., temporal edges that are used later in the sequence, are available later in time than the ones that come before~\cite{Whitbeck2012}. As we have seen in the example above, there exist cases where delaying some temporal edges makes some temporal paths ``faster''. In other words, both advancing and delaying some edges can make the existing infrastructure more efficient in terms of reachability times.

\subsection{Our contribution}


The problem studied, which we term \reachfast, can be summarized as follows. Given a temporal graph and a set of sources, {\em shift} some temporal edges, i.e., {\em advance} or {\em delay},  such that the maximum time any source needs to reach all the vertices is {\em minimized}. First, we study \reachfast with a single source in the graph. We distinguish between two cases: either some constraint on advances/delays is imposed, or no such constraints exist. We study two natural constrained versions that can appear in real life scenarios:
\begin{itemize}
    \item $\reachfast(k)$, where $k$ bounds the number of edges shifted;
    \item $\reachfasttotal(k)$, where $k$ bounds the sum of absolute shifts over all edges.
\end{itemize}
In Theorem~\ref{thm:one-source-hard} we prove that both constrained versions are \Wtwo-hard when parameterized by $k$, and \NP-hard otherwise. 
We complement these negative results with \iflong Algorithm~\ref{alg:one-source} \fi \ifshort an algorithm \fi that solves (unconstrained) \reachfast in polynomial time. Apparently, the problem becomes tractable the moment we remove the constraints on shifts. 

The next logical step, is to consider \reachfast for multiple sources, since networks may have multiple key vertices that need to reach the whole network. In Theorem~\ref{thm:two-sources-hard}, we prove that \reachfast becomes \NP-hard even if there are just two sources and the temporal graph is rather restricted: the graph has lifetime 1 and (static) diameter 6. This indicates that tractability requires some constraints on the input underlying graph. We derive efficient algorithms for \reachfast for several graph classes of the underlying network. 
For trees, we show that \reachfast can be solved in polynomial time for any number of sources. Our next result considers graphs of bounded treewidth.
We show that in this case, \reachfast can be encoded in Monadic Second Order logic, where the size of the formula depends only on the deadline by which we want to reach all the vertices from every source. Then, using Courcelle's Theorem, we get a fixed parameter tractable algorithm for the problem, parameterized by the treewidth of the graph and the lifetime of the optimal solution that works for any number of sources.
Finally, we consider a graph formed by several parallel paths with two common endpoints. For this class of graphs, we design a polynomial-time algorithm for \reachfast when there are only two sources and they are the endpoints of the paths.
Then, using Courcelle's Theorem, we get a fixed parameter tractable algorithm for the problem, parameterized by the treewidth of the graph and the lifetime of the optimal solution that works for any number of sources.

\paragraph{\textbf{Related work.}} The redesign of temporal graphs such that an objective is optimized, has received significant attention recently, mainly motivated by virus-spread minimization~\cite{braunstein2016,ENRIGHT201888}. A line of work studies minimization of reachability sets for a set of sources over a temporal graph using a variety of operations: delaying operations were first studied in~\cite{DP20}, alongside merging operations; edge-deletion and temporal edge-deletions operations were studied in~\cite{Deleting_Edges,COA2015}; reordering of the temporal edges was studied in~\cite{reordering}. Finally,~\cite{MRZ21} studies the relationship between the delaying and the edge-deletion operations under the reachability-minimization objective.
All of the above mentioned reachability problems heavily depend on the computation of temporal paths, a problem that has received a lot of attention~\cite{WC+14,WL+15,WH+16,shortest-experimental} and different definitions of temporal paths have been proposed~\cite{casteigts2021finding,thejaswi2020restless}. 
In ~\cite{li2018go} a similar idea to ours was explored under a slightly different objective. There, every temporal edge had a traversal-time that depends on the time it is used and  the goal was to find temporal paths that minimize the overall traversal-time. The authors observed that, for their objective, it might be beneficial to wait until they use an edge, since this might decrease the total traversal-time. 
Another recent work of interest \cite{KMMS22} studies the problem of finding the minimum number of labels that have to be added so that a temporal graph $G$ is temporally connected. This problem seems really close to ours, albeit they have some crucial differences. Namely, temporal connectivity requires ``both-ways'' connections, while we focus on ``one-way'' connections.

\ifshort
\smallskip

\noindent {\emph{Statements where proofs or details are omitted due to
space constraints are marked with $\star$. A version
containing all proofs and details is provided in the appendix.}}
\fi


\section{Preliminaries}
\label{sec:prelims}

For $n\in \mathbb{N}$, we denote $[n]:= \{1,2, \ldots, n\}$.
A {\em temporal graph} $\gcal := \tuple{G, \ecal}$ is defined by an {\em underlying graph} $G=(V,E)$ and a sequence of edge-sets $\ecal = (E_1, E_2, \ldots, E_{\tmax})$. It holds that $E = E_1\cup E_2\cup \dots \cup E_{t_{max}}$, $E_{t_{max}}\neq \emptyset$ and the {\em lifetime} of $\gcal$ is \tmax.
An edge $e \in E$ has label $i$, if it is available at time step $i$, i.e., $e \in E_i$. \iflong In addition, an edge has $k$ {\em labels}, if it appears in $k$ edge-sets. \fi The traversal time of an edge $e\in E_i$ at time step $i$ is $\tr(i,e)$; i.e., starting from one endpoint of the edge at time step $i$, we reach the other endpoint at the time step $i+\tr(i,e)$. 

A {\em temporal path} in $\tuple{G, \ecal}$ from vertex $v_1$ to vertex $v_m$ is a sequence of edges $P = (v_iv_{i+1}, t_i)_{i=1}^{m-1}$ such that for every $i \in [m]$ it holds that:
\begin{itemize}
    \item $v_iv_{i+1} \in E_{t_i}$, i.e. $v_iv_{i+1}$ is available at time step $t_i$;
    \item $t_i+\tr(t_i,v_iv_{i+1}) \leq t_{i+1}$ for every $i \in [m-1]$. 
\end{itemize}
The {\em arrival} time of $P$ is $t_m = t_{m-1}+\tr(t_{m-1},v_{m-1}v_{m})$. We call $P$ a foremost temporal $(v_1,v_m)$-path if 
for every $j\in [m-1]$ there is no temporal path $P'_j$ from $v_1$ to $v_{j+1}$ such that the arrival time of $P'_j$ is strictly smaller than the arrival time of $P_j = (v_iv_{i+1}, t_i)_{i=1}^{j}$.

A vertex $v$ is {\em reachable} from vertex $u$ by time $t$ in a temporal graph $\gcal$ if there exists a temporal path from $u$ to $v$ with arrival time at most $t$. We assume that a vertex is reachable from itself by time zero. It is possible that $u$ is reachable from $v$, but $v$ is not reachable from $u$. The reachability set of $v$, denoted $\reach(v, \gcal)$, contains all the vertices reachable from $v$. {\em The reaching-time} of $v$, denoted $\reachtime(v,\gcal)$, is $t$ if there exists a temporal path from $v$ to every vertex $u \in V$ with arrival time at most $t$. If there is a vertex $u \notin \reach(v,\gcal)$, then we will say that $v$ has infinite reaching-time.


The {\em shifting} of edge $uv \in E_i$ by $\delta \in \mathbb{Z}\setminus\{0\}$, denoted $\delay(uv,i) = \delta$, refers to replacing the label $i$ of this edge with the label $i+\delta$. When $\delta$ has a positive value, we say that the label is {\em delayed}, and when $\delta$ has a negative value, we say that the label is {\em advanced}. Note that shifting an edge is allowed only when $i+\delta\geq1$, but we allow $i+\delta> t_{max}$, i.e., we are allowed to increase the lifetime of the graph by delaying edges. Given a temporal graph $\gcal$ and a set of edges $X \subseteq \ecal$, we denote $\tilde{\gcal}(X)$ a class of temporal graphs where every edge in $X$ is shifted. A temporal graph $\gcal' \in \tilde{\gcal}(X)$ has $k$ edges shifted, if $|X|=k$; the total shift of $\gcal'$ is $k$ if $\sum_{i \in [\tmax]}\sum_{uv \in E_i} |\delay(uv,i)| = k$. Observe that shifting an edge can increase the reachability set of a vertex $v$, or it can decrease the reaching-time between two vertices. See Figure \ref{fig:prelims-example} for an example.


\begin{figure}[pos=h]
\begin{center}
\includegraphics[page=1]{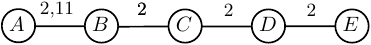}
\hspace{4cm}
\includegraphics[page=2]{prelims-example2_new.pdf}
\end{center}
\caption{Top: The original temporal graph of a network where we assume that every label has traversal time equal to $1$. Note that vertex $C$ reaches vertex $A$ at time step $12$ and vertex $E$ is unreachable. Bottom: The modified temporal graph by delaying edge $AB$ and advancing edge $CD$. Note that now, vertex $C$ reaches vertex $A$ at time step $4$ and vertex $E$ can now be reached at time step $3$.}
\label{fig:prelims-example}
\end{figure}

\noindent\textbf{The \reachfast problem.}
An instance of \reachfast consists of a temporal graph $\gcal := \tuple{G, \ecal}$ and a set of sources $S \subseteq V$.
The goal is to minimize the reaching-time any source requires to reach all the vertices of $V$ under shifting operations. We say that a temporal graph \(\gcal'\) is a \emph{solution} to an instance \(\tuple{\gcal= \tuple{G, \ecal}, S}\) of  \reachfast if \(\gcal'\in \tilde{\gcal}(X)\) for some \(X\subseteq E(G)\) and for every \(v\in S\) we have \(\reach(v,\gcal') = V(G)\). The \emph{value} of a solution \(\gcal'\) is \(\max_{v\in S}~ \reachtime(v, \gcal')\). We say that a solution \(\gcal'\) is \emph{optimal} if the value of \(\gcal'\) is minimized among all solutions. In other words, our objective is
$$
    \min_{X\subseteq V, \gcal'\in \tilde{\gcal}(X)} \max_{v\in S}~ \reachtime(v, \gcal').
$$

\noindent
In addition, we study two {\em constrained} versions of the problem.
In $\reachfast(k)$ we require that $|X| \leq k$, where $k$ is a part of the input; in $\reachfasttotal(k)$ we require that $\sum_{i \in [\tmax]}\sum_{uv \in E_i} |\delay(uv,i)| \leq k$. 




The next observation shows that when we can shift any number of edges, \reachfast under shifting reduces to \reachfast under delaying operations {\em only}.
\begin{observation} \label{obs:delaying}
Let  $\gcal := \tuple{G, \ecal}$ be a temporal graph, $S\subseteq V(G)$ a set of sources and let $\gcal' := \tuple{G, \ecal'}$ be the temporal graph constructed from $\gcal$ 
by advancing the labels of each edge $e$ to the first $k_e$ time steps, where $k_e$ is the number of labels of $e$. 
Then, any solution for \reachfast under $S$ in graph $\gcal'$ that only delays edges, is also a solution for \reachfast under $S$ in $\gcal$.
\end{observation}

\paragraph{\bf Parameterized complexity.}
We refer to the standard books for a basic overview of parameterized complexity theory~\cite{CyganFKLMPPS15,DowneyFellows13}.
At a high level, parameterized complexity studies the complexity of a problem with respect to its input size, $n$,  and the size of a parameter $k$. A problem is {\em fixed parameter tractable} by $k$, if it can be solved in time $f(k)\cdot \text{poly}(n)$, where $f$ is a computable function. 
A less favorable, but still positive, outcome is an $\XP{}$ \emph{algorithm}, which has running-time $\bigoh(n^{f(k)})$; problems admitting such algorithms belong to the class $\XP$. Showing that a problem is \Wh[t] rules out the existence of a fixed-parameter algorithm under the well-established assumption that \Wh[t]$\neq \FPT$.

\section{Single source}
\label{sec:one-source}
In this section we study the most basic case of the \reachfast problem, where there is only a single source.  To  begin with, we show that both constrained versions of the problem are \Wtwo-hard when parameterized by $k$. 
\iflong 
\begin{theorem}
\fi 
\ifshort 
\begin{theorem}[$\star$]
\fi 
\label{thm:one-source-hard}
$\reachfasttotal(k)$~and $ \reachfast(k)$ are \Wtwo-hard when~parameterized~by~$k$.
\end{theorem}
\ifshort
\begin{proof}[Proof sketch]
\fi
\iflong
\begin{proof}
\fi
We will prove the theorem via a reduction from \hittngset. 
An instance of \hittngset consists of a collection of subsets $S_1, S_2, \ldots, S_m$ over $[n]$ and a positive integer $k$, and we need to decide if there is a set $T \subset [n]$ of size $k$, such that $S_j \cap T \neq \emptyset$ for every $j \in [m]$. \hittngset is known to be \Wtwo-hard when parameterized by $k$ (see, e.g., \cite{CyganFKLMPPS15}).

Given an instance of \hittngset with \(m\) sets over \([n]\),we construct a temporal graph \gcal with $2n+m+1$ vertices as follows. 
First we let $u_0$ be the source. Then, we construct the vertices $x_1, x_2, \ldots, x_n$, and the vertices $y_1, y_2, \ldots, y_n$. For every $i \in [n]$ we have three edges: $u_0x_i$ with label 1; $x_iy_i$ with label 1; $u_0y_i$ with label 3. Finally, we construct the vertices $z_1, z_2, \ldots, z_m$, where vertex $z_j$ corresponds to subset $S_j$. If $i \in S_j$, then we create the edge $y_iz_j$ with label 3; Figure \ref{fig:one-source-constr} depicts our construction.
Observe that in \gcal, vertex $u_0$ reaches all $x$-vertices at time step 1 and all $y$-vertices at time step 3, but it does not reach any $z$-vertex. 

\begin{figure}[pos=h]

\includegraphics[scale=0.75]{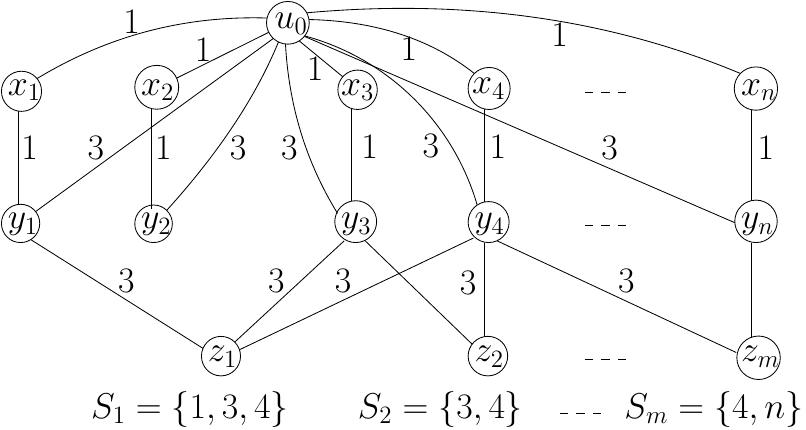}
\hspace{3.1cm}
\caption{Illustration of the construction used in Theorem~\ref{thm:one-source-hard}. 
Here, the subsets of the \hittngset instance are $S_1=\{1,3,4\},S_2,=\{3,4\},\ldots,S_m=\{4,n\}$.
}
\label{fig:one-source-constr}
\end{figure}

We claim that there exists a solution to $\reachfast(k)$ and $\reachfasttotal(k)$, such that $u_0$ reaches all the vertices at time step 3 that in the temporal graph we constructed above, if and only if there is a solution to the original \hittngset instance with $|T|=k$.
\iflong

Firstly, assume that we have a solution $T$ for the \hittngset instance. We construct $\tilde{\gcal}(X)$ as follows. For every $i \in T$, we add $x_iy_i \in X$ and in addition we delay it for one time step, i.e. we change the label of $x_iy_i$ from 1 to 2. Then observe the following. Vertex $u_0$ reaches all the $x$-vertices at time step 1. At time step 2 $u_0$ reaches every vertex $y_i$, where $x_iy_i \in X$, via the $x$-vertices; this is possible since all the edges in $X$ are available at time step 2. Finally, at time step 3 $u_0$ directly reaches the remaining $y$-vertices via the edges $u_0y_i$, where $i \notin T$, which are available at time step 3. In addition, for every $j\in [m]$, vertex $z_j$ is reached by $u_0$ via the edge $y_iz_j$ where $i \in T$. Observe that since $S_j \cap T \neq \emptyset$ for every $j \in [m]$, we get that by construction every $z$-vertex is connected to at least one $y$-vertex that is reachable from $u_0$ by time step 2.

For the other direction, assume that there is a solution to $\reachfast(k)$ (or $\reachfasttotal$ $(k)$) where $u_0$ reaches every vertex by time step 3 where the edges in $X$ are delayed. Then, observe that no edge with label 3 is delayed; any delay on these edges would not affect the solution. Similarly, it is not possible to advance any edge with label 1. In addition, clearly there is no reason to delay any edge between $u_0$ and a $x$-vertex. Hence, we conclude that the assumed solution of $\reachfast(k)$ delays only edges between $x$-vertices and $y$-vertices, or advances edges between $u_0$ and $y$-vertices. Furthermore, every edge in $X$ is delayed by one time step or advanced by one time step. This means that the total delay is $k$, hence the claimed solution for $\reachfast(k)$ is a solution for $\reachfasttotal(k)$ as well. We construct a solution for \hittngset as follows: for every $i \in [n]$, if edge $x_iy_i$ is delayed or if edge $u_0y_i$, then $i \in T$. Observe that $|T|=k$.
To see why $T$ is a solution for the \hittngset instance, observe that since $u_0$ reaches all the vertices by time step 3, we get that for every $j \in [m]$, vertex $z_j$ has to be adjacent to a vertex $y_i$ where $x_iy_i \in X$, i.e. $x_iy_i$ is delayed or $u_0y_i$ is advanced. Hence, by construction we get that for every $j \in [m]$ we get that $S_j \cap T \neq \emptyset$.
\fi
\end{proof}
\noindent
The proof of Theorem~\ref{thm:one-source-hard} immediately implies the following corollary. 
\begin{corollary}
\label{cor:one-source-np-h}
$\reachfast(k)$ and $\reachfasttotal(k)$ are \NP-hard when $k$ is part of the input.
\end{corollary}

The results above indicate that if we want to prove any positive result, then we need to either restrict the underlying graph structure, or study the unconstrained version of \reachfast. 
On the other hand, we can design a polynomial-time algorithm, termed Algorithm \ref{alg:one-source}, for \reachfast, when there is a single source.
We will utilize Observation~\ref{obs:delaying} and \ifshort wlog \fi \iflong without loss of generality \fi we will assume that all the labels appear at the first time step; thus we have to consider only delaying operations.
\iflong
Recall, due to Observation \ref{obs:delaying}, any solution given by our algorithms for a graph $G$ would also be a solution for every graph $G'$ that has the same vertex and edge set, and each edge has the same number of labels but with arbitrary values.
\fi

Algorithm \ref{alg:one-source} uses a greedy approach. In every round $r$, it checks every vertex $u$ that has been reached by the source by time $t_r$, and delays every incident unused temporal edge $uw$ to the time step with the minimum arrival time to $w$, if vertex $w$ has not yet been reached by source vertex $v$.


\iflong
\begin{theorem}
\fi
\ifshort
\begin{theorem}[$\star$]
\fi
\label{lem:oneSource}
Algorithm \iflong \ref{alg:one-source} \fi \ifshort  1 \fi  solves the unconstrained version of \reachfast with a single source in polynomial time.
\end{theorem}
\iflong
\begin{proof}
    We are going to prove that the algorithm solves the \reachfast problem by showing that the algorithm also minimizes the individual time it takes for the source vertex $v$ to reach any vertex $u\in V$.
    
    Consider any arbitrary vertex $u\in V$ after the execution of Algorithm \ref{alg:one-source}. Assume that the source vertex $v$ reaches vertex $u$ at time $t_i$ via path $P_1=(v,w_1,w_2,\ldots w_k,u)$ with a combination of edge delays, called $X$. Now assume that there is some other combination of edge delays, called $X'$ so that the source vertex $v$ reaches vertex $u$ at time $t_j<t_i$. This is possible either (i) via path $P_1$ with reaching-time $t_j$, or (ii) via some other path $P_2$ with reaching-time $t_j$. 
    
    Assume that reaching-time $t_j$ is achieved via path $P_1$. Note that if the optimal reaching-time for vertex $u$ is achieved through some combination of edge delays via path $P_1=(v,w_1,w_2,\ldots w_k,u)$, then the optimal reaching-time of every vertex $w_i\in P_1$ is achieved through that same path $P_1$ with the same combination of edge delays. Otherwise, if there was some other path $P_3$ that achieved a better reaching-time for some vertex $w_i\in P_1$, then the optimal reaching-time for vertex $u$ would be achieved via path $P_4=(P_3,w_i,w_{i+1},...u)$ which is impossible by assumption. 

    \smallskip
    \noindent\textbf{Case (i):} Consider now the optimal reaching-time of vertex $w_1$ via path $P_1$ and $X'$. Note that the algorithm delays the first available label of edge $uw_1$ to the time step which yields the minimum arrival time to $w_1$. Therefore $X'$ must contain the same edge delay for $uw_1$ as $X$. Consider now the optimal reaching-time of vertex $w_2$. Again, the single source algorithm delays the first available label of edge $w_1w_2$ to the time step such that yields the minimum arrival time to $w_2$ after $w_1$ has been reached. And since the reaching-time of $w_1$ is the same via $X$ and $X'$, the same is true for $w_2$. The argument works for every $w_i\in P_1$, including vertex $u$. Therefore, there is no other combination of edge delays that achieves a better reaching-time of vertex $u$ via path $P_1$.

    \smallskip    
    \noindent\textbf{Case (ii):} Consider now the optimal reaching-time of vertex $v$ via path $P_2=(v,w'_1,w'_2,\ldots w'_k,u)$ and $X''$ that the algorithm did not use. A similar argument works from the previous case that shows that if path $P_2$ had optimal reaching-time with $X''$, the single source algorithm would have used the edge delays of $X''$ and $X=X''$ for the edges on path $P_2$.
\end{proof}
\fi

\begin{algorithm}[t]
\caption{Single Source Reach Fast Algorithm}
\myalgo{myalgo:One-Source-Reach-Fast}
\vspace{0.2cm}
\label{alg:one-source}
\begin{algorithmic}[1]
\REQUIRE{A temporal graph $\gcal := \tuple{G, \ecal}$ and a source $v$.}
\ENSURE{An optimal solution $\gcal'$.}

\STATE{$\mathcal{S} \leftarrow v$; $\mathcal{S}_{t} \leftarrow \emptyset$; $r=0$;}

\WHILE{$|\mathcal{S}|\neq|V|$}
    \FOR{every vertex $u\in \mathcal{S}$}
        \FOR{every neighbor $w\notin \mathcal{S}$ of $u$}
            \FOR{$i=1,2,\ldots, r$}
                \IF{$uw\in E_i$}
                    \STATE{$min=i+\tr(i,uw)$}
                    \FOR{$j==i,j++,min<j$}
                        \IF{$min<\tr(j,wu)$}
                            \STATE{$min=\tr(j,wu)$}
                        \ENDIF    
                    \ENDFOR 
                    \STATE{$\delay(uw,i) = j-i$}
                    \STATE{$\mathcal{S}_{t} \leftarrow \mathcal{S}_{t}\cup \{w\} $} 
                \ENDIF
            \ENDFOR
        \ENDFOR    
    \ENDFOR
    \STATE{$\mathcal{S} \leftarrow \mathcal{S}\cup \mathcal{S}_{t}$; $\mathcal{S}_{t} \leftarrow \emptyset$; $r++$;}
\ENDWHILE

\end{algorithmic}
\end{algorithm}

\section{Two sources}
\label{sec:two-source}
In this section we study the case where we have multiple sources. Since  we know from Theorem~\ref{thm:one-source-hard} that the constrained versions of \reachfast are \NP-hard even when there is a single source, we will only consider the unconstrained case. As we will show below, the unconstrained version of \reachfast becomes hard even if there are just two sources. In fact, the underlying graph that results from our reduction always admits an optimal solution of a constant value. Note that this also means that the problem is even \APX-hard; i.e., unless $\PP = \NP$, there is no polynomial-time \(c\)-approximation for some \(c>1\). We highlight that the constructed temporal graph has lifetime 1 and every label has traversal time 1 as well. This means that, essentially, there are no temporal constraints in the edges.

\iflong
\begin{theorem}
\fi
\ifshort
\begin{theorem}[$\star$]
\fi
\label{thm:two-sources-hard}
It is \NP-complete to decide whether \reachfast with two sources admits a solution of value 6, even on temporal graphs of lifetime 1.
\end{theorem}

\iflong
The \reachfast problem belongs to \NP since given a solution of value $p$, we can check whether there is a temporal path from each source to every vertex with arrival time smaller or equal to $p$.
\fi
We will prove the theorem via a reduction from a special version of \naesat, termed \monlinnaesat~\cite{DD20}. \iflong An instance of \monlinnaesat consists of a \fi \ifshort We are given a  \fi  CNF formula $\phi$ that involves $n$ variables \iflong $x_1, x_2, \ldots, x_n$ \fi and $m$ clauses \iflong $c_1, c_2, \ldots, c_m$ \fi that satisfy the following constraints:
\iflong
\begin{itemize}
    \item every clause is made up of exactly three distinct literals;
    \item each variable appears exactly four times;
    \item there are no negations in the formula, i.e. it is monotone;
    \item the formula is linear; i.e. each pair of clauses shares at most one variable.
\end{itemize}
\fi
\ifshort
every clause is made up of exactly three distinct literals;
each variable appears exactly four times;
there are no negations in the formula;
each pair of clauses shares at most one variable.
\fi
The goal is to decide whether there exists a truth assignment \iflong for the variables\fi, such that for each clause at least one literal is set to true and at least one literal is set to false.

\paragraph{\bf Construction.} 
Given an instance $\phi$ of \monlinnaesat, we will create a temporal graph \gcal with $3n+4m+12$ variables. Since the constructed graph has lifetime 1, we do not have to specify the labels of the edges. The two sources are the vertices $s$ and $s'$. For every clause $c_j$ we create two clause-vertices $a_j, a'_j$, two auxiliary vertices $z_j, z'_j$, and the edges $sa_j, s'a'_j$ and $a_jz_j, a'_jz'_j$. For every variable $x_i$, we create the vertices $b_i, b'_i$ and $w_i$. If variable $x_i$ appears in clause $c_j$, then we add the edges $a_ib_j$ and $a'_ib'_j$. In addition if variables $x_i$ and $x_j$ belong to the same clause we add the edges $b_ib_j, b'_ib'_j$ and $w_iw_j$. Finally, we create a ``ladder'' between $s$ and $s'$ using the vertices $p_i, q_i$ where $i \in [5]$. The ``ladder'' is created by the edges $sp_1, sq_1, s'p_5, s'q_5$ and the edges $p_iq_i$ for every $i \in [5]$. Figure~\ref{fig:two-sources-constr} depicts our construction. Theorem~\ref{thm:two-sources-hard} follows immediately from the following two lemmas.


\begin{figure}[pos=h]
\includegraphics[scale=0.7]{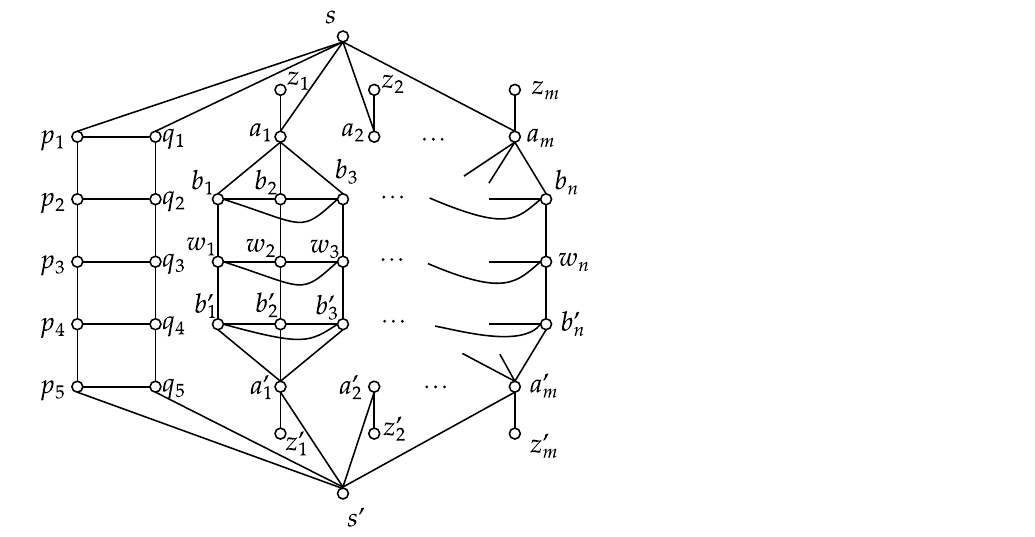}
\caption{The construction used in Theorem~\ref{thm:two-sources-hard}. The first clause of the \monlinnaesat instance is $c_1 =(x_1 \vee x_2 \vee x_3)$.}
\label{fig:two-sources-constr}
\end{figure}

\iflong
\begin{lemma}
\fi
\ifshort
\begin{lemma}[$\star$]
\fi
\label{lem:two-sound}
If the instance of \monlinnaesat is satisfiable, then for the constructed temporal graph \gcal we have that $\min_{X\subseteq V, \gcal'\in \tilde{\gcal}(X)} \max_{v\in S}~ \reachtime(v, \gcal') = 6$.
\end{lemma}
\iflong
\begin{proof}
\fi
\ifshort
\begin{proof}[Proof sketch]
\fi
\iflong Assume that we have a truth assignment for the variables $x_1, x_2, \ldots, x_n$ that satisfies the instance of \monlinnaesat.\fi~ We create the \iflong following \fi labelling for \gcal whose pattern is depicted in Figure~\ref{fig:two-sources-sound}. \iflong 
\begin{itemize}
    \item For every $i \in [m]$, every edge has $\tr(i,e)=1$.
    \item For every $i \in [m]$ the edges $sa_i$ and $s'a'_i$ have label 1.
    \item If variable $x_i = \mathtt{True}$ and it appears at clause $c_j$, then: edge $a_jb_i$ has label 2; edge $b_iw_i$ has label 3; edge $w_ib'_i$ has label 4; edge $b'_ia'_i$ has label 5.
    \item If variable $x_i = \mathtt{False}$ and it appears at clause $c_j$, then: edge $a_jb_i$ has label 5; edge $b_iw_i$ has label 4; edge $w_ib'_i$ has label 3; edge $b'_ia'_i$ has label 2.
    \item Every edge between any two $b$-vertices, any two $b'$-vertices, or any two $w$-vertices has label 5.
    \item For every $i \in [m]$ the edges $a_iz_i$ and $a'_iz'_i$ have label 6.
    \item Edge $sp_1$ has label 1; edge $p_1p_2$ has label 2; edge $p_2p_3$ has label 3; edge $p_3p_4$ has label 4; edge $p_4p_5$ has label 5; edge $p_5s'$ has label 6.
    \item  For the ``ladder'' he have the following: edge $s'q_5$ has label 1; edge $q_5q_4$ has label 2; edge $q_4q_3$ has label 3; edge $q_3q_2$ has label 4; edge $q_2q_1$ has label 5; edge $q_1s$ has label 6. Finally, for every $i \in [5]$, the edge $p_iq_i$ has label 6.
\end{itemize}
\fi

\begin{figure}[pos=h]

\includegraphics[scale=0.7]{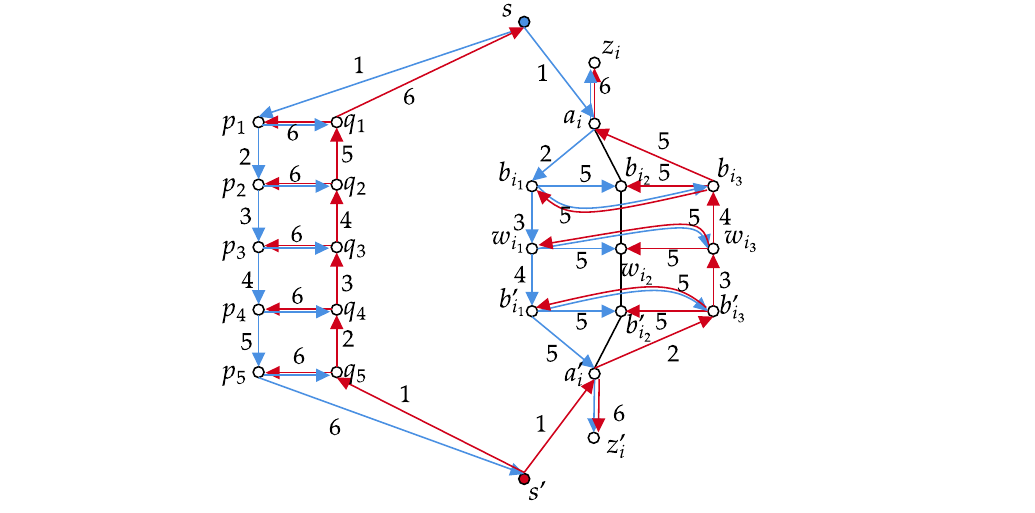}
\hspace{2.0cm}
\caption{The labelling of $\tilde{\gcal}(X)$ for a given satisfying assignment of \monlinnaesat, where in clause $c_j = (x_{i_1} \vee x_{i_2} \vee x_{i_3})$ we have $x_{i_1} = \mathtt{True}$ and $x_{i_3} = \mathtt{False}$; the assignment of $x_{i_2}$ does not affect the reachability. The coloured directed paths, are the temporal paths from the corresponding source. The edges with two colors and two arrows are used at the same time step by both colours.}
\label{fig:two-sources-sound}
\end{figure}

\ifshort
The proof is completed by observing that under this labelling scheme $s$ and $s'$ reach every vertex by time step 6.
\fi
\iflong
To see that under this labelling scheme $s$ and $s'$ reach every vertex by time step 6, observe the following. It is straightforward to see from Figure~\ref{fig:two-sources-sound} that every vertex of the ``ladder'' and the vertices $s$ and $s'$ are reachable by both sources by time step 6. For the remaining vertices, let us focus on clause $c_i = (x_{i_1} \vee x_{i_2} \vee x_{i_3})$ and let $V_i = \{a_i, a'_i, z_i, z'_i, b_{i_1}, b'_{i_1}, w_{i_1}, b_{i_2}, b'_{i_2}, w_{i_2}, b_{i_3}, b'_{i_3}, w_{i_3}\}$.
\begin{itemize}
    \item If $x_{j} = \mathtt{True}$ for some $j \in \{i_1, i_2, i_3\}$, then $a'_i$ is reached by $s$ at time step 5. Hence, according to Figure~\ref{fig:two-sources-sound}, by time step 6 all vertices in $V_i$, will be reached by $s$.
    \item If $x_{j} = \mathtt{False}$ for some $j \in \{i_1, i_2, i_3\}$, then $a_i$ is reached by $s'$ at time step 5. Hence, according to Figure~\ref{fig:two-sources-sound}, by time step 6 all vertices in $V_i$, will be reached by $s$.
\end{itemize}
Since we start from a solution for \monlinnaesat, we know that for every $i \in [m]$ there exist $x_{j} = \mathtt{True}$ and $x_{j'} = \mathtt{False}$ with $j, j' \in \{i_1, i_2, i_3 \}$. Thus, for every $i \in [m]$, all vertices of $V_i$ are reached by both sources by time step 6.
\fi
\end{proof}
\iflong
\begin{lemma}
\fi
\ifshort
\begin{lemma}[$\star$]
\fi
\label{lem:two-source-correct}
If the constructed temporal graph \gcal admits a solution of value at most $6$,
then there is a satisfying assignment for the instance of \monlinnaesat.
\end{lemma}
\iflong
\begin{proof}
Assume that there is a labelling on the edges of \gcal such that  every vertex of \gcal is reachable by time step 6 by both sources. Let $V_i = \{a_i, a'_i, z_i, z'_i, b_{i_1}, b'_{i_1}, w_{i_1}, b_{i_2}, b'_{i_2}, w_{i_2}, b_{i_3}, b'_{i_3}, w_{i_3}\}$ for some $i \in [m]$; $V_i$ contains the vertices related to clause $c_i = (x_{i_1} \vee x_{i_2} \vee x_{i_3})$. Let us focus on the vertices $z_i$ and $z'_i$. Observe that the distance between $s$ and $z'_i$, and the distance between $s'$ and $z_i$ in the underlying graph of \gcal is 6. In addition, observe that any static path between $s$ and $z'_i$ in the underlying graph that has length 6 involves vertices {\em only} from $V_i$. The same holds for the vertices $s'$ and $z_i$. Hence, for every $i \in [m]$ the following two temporal paths have to exist:
\begin{itemize}
    \item $P_{i} = \big((sa_{i},1), (a_{i}b_{i_1},2), (b_{i_1}w_{i_1},3),  (w_{i_1}b'_{i_1},4)$, $ (b'_{i_1}a'_{i_1},5), (a'_{i}z'_{i_1},6)\big)$.
    \item $P'_{i} = \big((s'a'_{i},1), (a'_{i}b'_{i_2},2), (b'_{i_2}w_{i_2},3),  (w_{i_2}b_{i_2},4)$, $ (b_{i_2}a_{i_2},5),  (a_{i}z_{i_1},6)\big)$.
\end{itemize}
Since every edge of \gcal has only one label, it must be true that $i_1 \neq i_2$.
So, given this labelling scheme, we construct a truth assignment for $x_1, x_2, \ldots, x_n$ by setting $x_i = \mathtt{True}$ if the edge $b_{i}w_{i}$ has label 3 and $x_i = \mathtt{False}$ if the edge $b_{i}w_{i}$ has label 4. Since the edges of \gcal have only one label, we get that the truth assignment is valid. To see that this assignment is a solution for \monlinnaesat, observe that each $a$-vertex belongs to at least one $P$-path and at least one $P'$-path, hence by our construction of the truth assignment the corresponding clause for the $a$-vertex has at least one variable that is True and at least one variable that is False.
\end{proof}
\fi
\section{Tractable Cases}
\label{sec:upper-bounds}

In this section, we give 
algorithms for \reachfast for three special classes of underlying graphs: trees, parallel paths, and graphs of bounded treewidth. 
Due to Observation~\ref{obs:delaying}, we assume that all the labels of an edge form an interval starting at time step one, hence we will study {\em only} delaying operations.

\subsection{Trees}

\newcommand{\sdiam}{\mathtt{s}\xspace}
\newcommand{\sol}{\mathtt{\gcal}\xspace}

In this section, we consider temporal graphs with multiple sources where the underlying graph is a tree. To devise a polynomial-time algorithm, we want to limit the number of solutions we have to check in order to find an optimal one.  We start with an observatio that motivates the definition that follows it.

\begin{observation}
    Consider a tree graph $\gcal$ and a set of sources $S$. Let $t$ be the earliest time step for which there exists at least one vertex that is reached by all sources. Then, at time step $t$ there is exactly one vertex $u$ that has been reached by both sources or there are exactly two vertices $u,v$ that have been reached by all the sources, and $u,v$ are neighboring vertices.
\end{observation}
The previous observation motivates us to group solutions based on the first vertices that are reached by all sources. 
 
\begin{definition}
    Consider a tree graph $\gcal$ and a set of sources $S$. We define $B_u$ to be the set of solutions for the \reachfast problem where $u$ is the first vertex to be reached by every source and no other vertex has been reached by every source by the same time. We also define $B_{uv}$ to be the set of solutions where vertices $u,v$ is the first pair of adjacent vertices to be reached by every source. Note that we require $u,v$ to be reached at exactly the same time step by the last sources that reach them.
\end{definition}

\begin{observation}\label{obs:labelsPerEdge}
    Consider a tree graph $\gcal$ with a set of sources $S$. If there exists a path between two sources where at least two edges have a single label, then there is no solution to the \reachfast problem. If there exists a path between two sources where exactly one edge has a single label, then set $B_u$ is empty.
\end{observation}

We will provide an algorithm that computes the best solution of each set in polynomial time and then picks the best solution among the sets for the optimal solution. 
We first show how our algorithm computes a solution for a set $B_u$ with a single first vertex $u$ and then extend to the solutions for sets $B_{u,v}$ with a pair of adjacent vertices. The algorithm computes a solution $\sol_u$ for the \reachfast problem, where vertex $u$ is the first vertex to be reached by every source in two phases as follows: Let $\sdiam$ be the maximum distance between vertex $u$ and every source $v\in S$. We root the tree to vertex $u$ and put every vertex $v$ in layer $\ell$ if it has distance $\ell$ from the root. In the first phase, we delay the first label on each edge between layer $i$ and $i+1$ to value $\sdiam-i$. If that operation would give a $0$ or negative value, we do not perform the delay. This relabeling guarantees that in the current graph, every source reaches $u$ at time step $\sdiam+1$. In the second phase, we delay the second label on each edge between layer $i$ and $i+1$ to value $\sdiam+i+1$. See Figure \ref{fig:algorithm1} for an illustration. Note that Observation \ref{obs:labelsPerEdge} guarantees that there are enough labels to guarantee that the relabelling is correct

To compute the solution $\sol_{u,v}$ for the \reachfast problem where vertices $u,v$ are the first vertices to be reached by every source at exactly the same time step, the algorithm does the following. We virtually cut edge $uv$ and we now have 2 connected trees $G_u$ and $G_v$. Let $\sdiam_u$ be the maximum distance between vertex $u$ and every source in $G_u$ and let $\sdiam_v$ be the maximum distance between vertex $v$ and every source in $G_v$. Let $\sdiam$ be the maximum between $\sdiam_u+1$ and $\sdiam_v+1$. We root $G_u$ at edge $u$ and $G_v$ at edge $v$ and put every vertex in a layer based on its distance from its respective root. In the first phase, we delay the first label on each edge between layer $i$ and $i+1$ to value $\sdiam-i$ in both $G_u$ and $G_v$. If that operation would give a $0$ or negative value, we do not perform the delay. We also delay the first label of edge $uv$ to $\sdiam+1$. In the second phase, we delay the second label on each edge between layer $i$ and $i+1$ to value $\sdiam+i+2$. This completes the description. We will now show that solution $\sol_u$ is better than all other solutions in $B_u$ and solution $\sol_{u,v}$ is better than all other solutions in $B_{u,v}$.

\begin{figure}[pos=h]
\begin{center}
\includegraphics[page=2,scale=0.55]{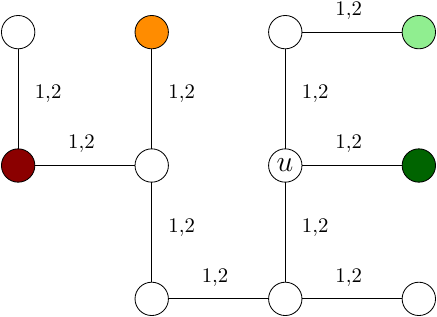}
\hspace{1.5cm}
\includegraphics[page=3,scale=0.55]{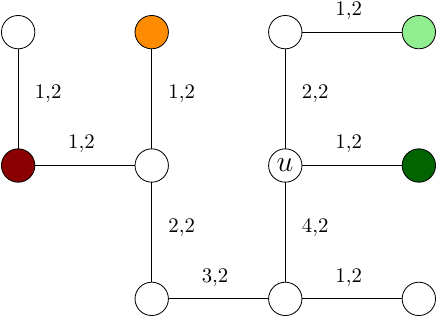}
\hspace{1.5cm}
\includegraphics[page=4,scale=0.55]{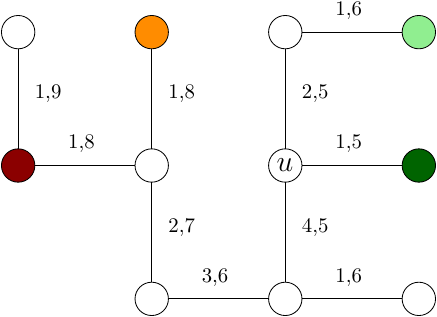}
\end{center}
\caption{An example of the algorithm running on a tree graph with 4 sources for set $B_u$. The left figure shows the input graph. The middle figure shows the graph after phase 1 has changed the labels. The right figure shows the graph after phase 2 has changed the labels.}
\label{fig:algorithm1}
\end{figure}

\begin{figure}[pos=h]
\begin{center}
\includegraphics[page=2,scale=0.55]{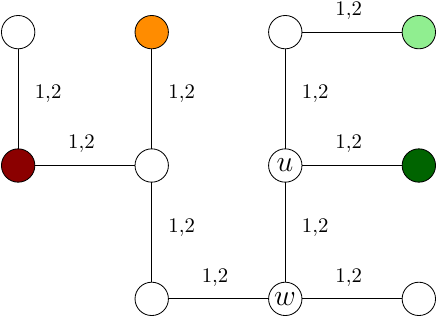}
\hspace{1.5cm}
\includegraphics[page=3,scale=0.55]{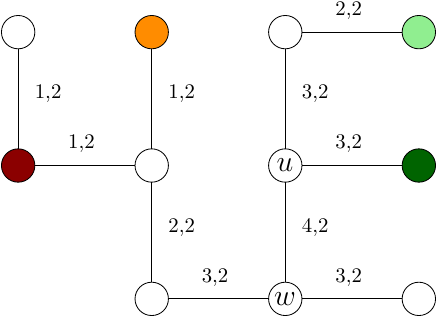}
\hspace{1.5cm}
\includegraphics[page=4,scale=0.55]{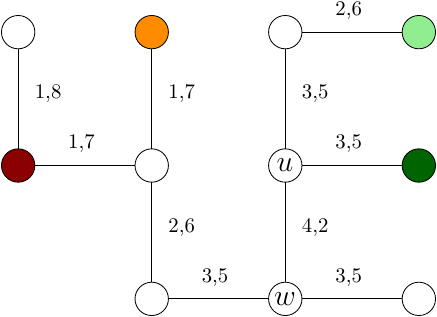}
\end{center}
\caption{An example of the algorithm running on a tree graph with 4 sources for set $B_{u,w}$. The left figure shows the input graph. The middle figure shows the graph after phase 1 has changed the labels. The right figure shows the graph after phase 2 has changed the labels.}
\label{fig:algorithm1}
\end{figure}

\begin{lemma}\label{lem:optimal}
    Let $\gcal$ be a tree graph where every edge has at least two labels and $\sol_u,\sol_{u,w}$ be the solutions constructed by the algorithm. For every solution $\gcal'\in B_u$, $\max_{v\in S}~ \reachtime(v, \sol_u)\le \max_{v\in S}~ \reachtime(v, \gcal')$ and for every solution $\gcal''\in B_{u,v}$, $\max_{v\in S}~ \reachtime(v, \sol_{u,v})\le \max_{v\in S}~ \reachtime(v, \gcal'')$.
\end{lemma}

\begin{proof}
    Let us first analyze the solution $\sol_u$ given by the algorithm. Let $\depth$ be the depth of tree $\gcal$ if we root it at $u$, and let $\sdiam$ be the maximum distance between vertex $u$ and every other source $v\in S$. The value of the solution is at most $\sdiam+\depth+2$, since every source can reach vertex $u$ by time step $\sdiam+1$, and then every other vertex can be reached from $u$ by time step $\depth+1$. Our goal is to show that every other solution in $B_u$ is at least $\sdiam+\depth+2$ which would give us a tight bound on the best solution. Let $sink_{u,\gcal'}$ be the minimum time step by which all sources have reached vertex $u$ in a solution $\gcal'$. We will show that in every solution in $B_u$: (i) once the last source $v$ reaches $u$, $v$ needs at least $\depth+1$ time steps to reach at least one vertex it has not reached already and (ii) $sink_{u,\gcal'}\geq \sdiam+1$.
    
    First note that in every solution in $B_u$ vertex $u$ is reached by all of the sources at time step $\sdiam+1$ or afterwards. This holds because $\sdiam$ is equal to the static distance between some source $v$ and $u$ in $\gcal$. Therefore $sink_{u,\gcal'}\geq \sdiam+1$ and claim (ii) holds true.

    Let $t$ be the time step that the last sources $v_1,v_2,\ldots,v_k$ reaches $u$ in some solution in $B_u$. There exists some vertex $z$ with static distance $\depth$ from $u$. We will argue that the following is true for at least one of the sources $v_1,v_2,\ldots,v_k$: the closest vertex to $z$ that the source $v_i$ has reached is $u$. This implies that $v$ needs at least $\depth+1$ time steps to reach $z$ which proves claim (ii). Consider path $P_{uz}=(u,x_1,x_2,\ldots,x_l,z)$. We will argue that some source $v_i$ has not yet reached vertex $x_l$ by time step $t$. Assume otherwise, that every source $v_1,v_2,\ldots,v_k$ has reached vertex $x_l$ by time step $t$. This implies that the first vertex to be reached by every source is vertex $x_l$ and not vertex $u$, or vertices $x_l$ and $u$ were the first vertices to be reached at the same time. Both cases lead to a contradiction. 

    The proof for solution $\sol_{u,w}$ works in a similar fashion. 
\end{proof}

\begin{theorem}
    \reachfast can be solved in polynomial time on tree graphs.
\end{theorem}

\begin{proof}
    By Lemma \ref{lem:optimal} and we can compute the best solution for sets $B_u$ and $B_{u,v}$ for every vertex of the graph. By Observation \ref{obs:labelsPerEdge} we know that the union of all sets contains every possible solution. Thus, computing the best solution for these $2n-1$ sets gives us the optimal solution to the \reachfast problem in polynomial time.
\end{proof}

\subsection{Bounded-Treewidth Graphs}

In this subsection, we show that \reachfast\ admits an \FPT\ algorithm parameterized by treewidth plus the value of the sought solution $D$; i.e., when looking for a solution \(\gcal'\) such that for all $v\in S$ it holds that $ \reachtime(v, \gcal')\le D$. $D$ can be viewed as a \emph{deadline} by which we have to reach all vertices of the graph from all sources. 

\paragraph{\bf Treewidth.}
A \emph{tree-decomposition}~$\mathcal{T}$ of a graph $G=(V,E)$ is a pair 
$(T,\chi)$, where $T$ is a tree (whose vertices we call \emph{nodes}) and $\chi$ is a function that assigns each node $t$ a set $\chi(t) \subseteq V$ such that the following conditions hold.
\begin{itemize}[noitemsep]
	\item For every $uv \in E$, there is a node	$t$ such that $u,v\in \chi(t)$.
	\item For every vertex $v \in V$,
	the set of nodes $t$ satisfying $v\in \chi(t)$ forms a subtree of~$T$.
\end{itemize}
\noindent 
The \emph{width} of a tree-decomposition $(T,\chi)$ is the size of a largest set $\chi(t)$ minus~$1$, and the \emph{treewidth} of the graph $G$,
denoted $\tw(G)$, is the minimum width of a tree-decomposition of~$G$.

\paragraph{ \bf Monadic Second Order Logic.} We consider \emph{Monadic Second Order} (MSO) logic on labeled graphs in terms of their incidence structure whose universe contains vertices and
edges; the incidence between vertices and edges is represented by a
binary relation. \iflong We assume an infinite supply of \emph{individual
 variables} $x,x_1,x_2,\dots$ and of \emph{set variables}
$X,X_1,X_2,\dots$\ The \emph{atomic formulas} are 
$V x$ (``$x$ is a vertex''), $E y$ (``$y$ is an edge''), $I xy$ (``vertex $x$
is incident with edge $y$''), $x=y$ (equality), $x\neq y$ (inequality),
$P_a x$ (``vertex or edge $x$ has label $a$''), and $X x$ (``vertex or
edge $x$ is an element of set $X$'').  \emph{MSO formulas} are built up
from atomic formulas using the usual Boolean connectives
$(\lnot,\land,\lor,\rightarrow,\leftrightarrow)$, quantification over
individual variables ($\forall x$, $\exists x$), and quantification over
set variables ($\forall X$, $\exists X$).

\fi 
\smallskip
\noindent
Let $\Phi(X_1,X_2,\ldots, X_k)$ be an MSO formula with a free set variables $X_1,\ldots, X_k$. For a labeled
graph $G=(V,E)$ and sets $S_1,\ldots, S_k\subseteq E$ we write $G \models \Phi(S_1,\ldots, S_k)$ if
the formula $\Phi$ holds true on $G$ whenever $X_i$ is instantiated with~$S_i$ for all \(i\in [k]\).

\smallskip
\noindent
\iflong The following result (an extension of the well-known Courcelle's Theorem~\cite{Courcelle90}) will be useful to show that \reachfast\ admits an FPT algorithm parameterized by treewidth and the value of the sought solution.
\fi 

\begin{fact}[\cite{ArnborgLagergrenSeese91}]\label{fact:MSO} 
  Let $\Phi(X_1,\ldots, X_k)$ be an MSO formula with a free sets variables $X_1,\ldots, X_k$ and $w$ a
  constant. Then there is a linear-time algorithm that, given a labeled
  graph $G=(V,E)$ of treewidth at most $w$, 
  decides whether there exists $k$ sets $S_1,\ldots, S_k,\subseteq E(G)$ such that $G \models \Phi(S_1,\ldots, S_k)$. Moreover, if such solution exists, the algorithm constructs one such solution. 
\end{fact}
We note that \cite{ArnborgLagergrenSeese91} does not explicitly construct the solution, however, for each formula, they construct a tree automaton that works along the decomposition and stores $f(w, |\Phi(X_1,\ldots, X_k)|)$-many records (for some computable function $f$) for the optimal solutions. To make it constructive, for each record we only need to remember a viable representative.
\iflong 
Furthermore, there is also requires a
tree-decomposition of width at most $w$ to be provided with the
input. However, 
for a graph of treewidth at most $w$ such a tree
decomposition can be found in linear time~\cite{Bodlaender96}, hence we can
drop this requirement from the statement of the theorem.
\fi

As in the previous sections, we again assume that we are only allowed to delay the edges, which we can assume without loss of generality due to Observation~\ref{obs:delaying}. \iflong The main result of this subsection is the following theorem. \fi
 \iflong
 \begin{theorem}
 \fi 
 \ifshort
 \begin{theorem}[$\star$]
 \fi \label{thm:tw+deadline}
 	Given {a temporal graph} $\gcal := \tuple{G, \ecal}$ such that $G$ has treewidth bounded by a constant, a set of sources $S\subseteq V(G)$, and a constant $D$, there is a linear-time algorithm that either computes a solution $\gcal'$ of value at most $D$ or decides that the value of an optimal solution is~at~least~$D+1$. 
 \end{theorem} 
\ifshort
\begin{proof}[Proof Sketch]
\fi
\iflong
\begin{proof}
\fi     
	To prove the theorem we construct a labeled graph $H$ such that treewidth of $H$ is at most treewidth of $G$ and an MSO formulation $\Phi(\mathbf{X})$ such that the size of formula $\Phi(\mathbf{X})$ depends only on $D$ and $\mathbf{X} = (X_1^1, X_1^2,\ldots X_1^D, X_2^1,X_2^2, \ldots X_D^D)$. 
	\iflong Before we start our construction, o\fi\ifshort O\fi bserve that if for some edge $e\in E(G)$ and some time step $i$, we have $\tr(i,e)>D$, then we will never delay the edge $e$ to the time step $i$\iflong , as this would not help us to reach any new vertex by the time $D$. Therefore, for the sake of presentation of the proof, we will \fi\ifshort{} and we can \fi assume that whenever $\tr(i,e)>D$, then $\tr(i,e)=D+1$. 
	For sets \(S_1, \ldots, S_D\) of edges, let \(\gcal[S_1, \ldots, S_D]\) denote the temporal graph obtained from $\gcal$ by delaying, for each $i\in [D]$, every edge in $S_i$ to $E_i$. Moreover, for every $i\in [D]$ let $\{S_i^1,\ldots, S_i^D\}$ be the partition of $S_i$ such that edge $e\in S_i$ belongs to $S_i^j$ if and only if $\tr(i,e)=j$. We will show that \(H\models \Phi(S_1^1,\ldots, S_D^D)\) if and only if in the temporal graph \(\gcal[S_1, \ldots, S_D]\) it holds for all $v\in S$ that $\reachtime(v, \gcal[S_1, \ldots, S_D])\le D$. The result then follows by Fact~\ref{fact:MSO}. 
	To construct $H$ we let
	\begin{enumerate}
		\item $V(H)=V(G)$,
		\item $E(H)= \bigcup_{i\in [D]}E_i$ (note that $E(H)\subseteq E(G)$), 
		\item the edge $e\in E(H)$ has label $(I,T)$, where $I\subset [D]$ and $T = (t_1,\ldots, t_D)\in [D+1]^D$ such that \iflong \begin{itemize}
		    \item for all $i\in [D]$ we have $i\in I$ if and only if $e\in E_i$;
		    \item for all $j\in [D]$, we have $\tr(j,e)=t_j$.
		\end{itemize} 
		\fi\ifshort (1) for all $i\in [D]$ we have $i\in I$ if and only if $e\in E_i$, and (2) for all $j\in [D]$, we have $\tr(j,e)=t_j$,
		\item every vertex in $S$ has label $S$.\fi
	\end{enumerate}

\newcommand{\checkEdge}{\operatorname{allowed\_delay}}
\newcommand{\checkTraversal}{\operatorname{traversal\_time}}
\newcommand{\checkPath}{\operatorname{path}}
\newcommand{\isEdgeSet}{\operatorname{is\_edge\_set}}
\ifshort 
It is easy to see that $H$ is a subgraph of $G$ and hence the treewidth of $H$ is bounded by the treewidth of $G$. The MSO formulation $\Phi(\mathbf{X})$ is a conjunction of four building blocks (subformulas) that can each be easily expressed as an MSO formulation of the size depending only on $D$. The first block check that each of the sets $X^j_i$ is an edge set. The second block checks that every edge in $X_i^j$ has a label $(I,T)$ with $t_i=j$. The third block checks that if some edge is precisely at the time steps in $I\subseteq [D]$ after the delaying, that is the edge appears in some $X_i^j$ for every $i\in I$ and no other $X_i^j$, then we can delay this edge to these time steps. This is simply a big conjuction of small subformulas stating that if an edge is precisely in some given $X_i^j$'s then the edge has one of the labels that allow for these $X_i^j$'s. 
Finally, the fourth block verifies that for every source vertex $s$ and every other vertex $t$ in the graph, there is an $s$-$t$ path using only edges in some $X_i^j$'s. Note that each such path has to have length at most $D$ and each edge on the path should belong to one of $X_i^j$. Moreover, such temporal path is valid exactly when for two consecutive edges $e_1\in X_{i_1}^{j_1}$, $e_2\in X_{i_2}^{j_2}$, we have $i_1+j_1\le i_2$. Hence the fourth block is a disjunction over all less than $D^3$ many valid possibilities of an $s$-$t$ path that reaches the target by time $D$ and checking that there exists an $s$-$t$ path with edges in the prescribed sets $X_i^j$. 
\end{proof}
\fi

\iflong
    Note that if $e$ has label $(I,T)$, then $T = (\tr(1,e), \tr(2,e),\ldots, \tr(D,e))$. For convenience, for the rest of the proof we will denote by $T_e$ the vector $(\tr(1,e), \tr(2,e),\ldots, \tr(D,e))$. Note that $T_e$ is fixed for given edge $e$ and does not depend on $I$.

	Now we construct the formula  $\Phi(\mathbf{X})$. 
	First, we need check that if the algorithm from Fact~\ref{fact:MSO} outputs $(S_1^1,\ldots, S_D^D)$ such that the edge $e$ appears in sets $S_{i^e_1}^{t^e_1}, \ldots, S_{i^e_q}^{t^e_q}$, for some $q\in [D]$, then $e$ appears in at least $q$ time steps and can be delayed to time steps ${i^e_1}, \ldots, {i^e_q}$ and that for all $x\in q$ we have $\tr(i^e_x, e) = t^e_x$. For two sets of indices $I,J\subseteq [D]$, we say that $I$ \emph{dominates} $J$, denoted $J\preceq I$, if $I = \{i_1,\ldots, i_q\}$ such that $i_1< i_2\cdots< i_q$, $J = \{j_1,\ldots, j_p\}$ such that $j_1< j_2\cdots< j_p$, $q\le p$, and for all $r\in [q]$ we have $j_r\le i_r$. Note that if we can delay the edge $e$ to appear at time steps in the set $I$, then the edge $e$ has in $H$ some label $(J,T_e)$ such that $I$ dominates $J$. 
	 We define an auxiliary formula that takes as input an edge $e$, $I\subseteq [D]$
	\begin{align*}
	\checkEdge & (e,I) :=  \\ & \bigwedge_{T=(t_1,\ldots, t_D) \in [D+1]^D}\left(\Big(\bigvee_{J\subset [D]}P_{(J,T)}e\Big)\rightarrow
	    \Big(\big(\bigwedge_{i\in I}X_i^{t_i}e\land \bigwedge_{i\notin I\wedge j\in [D]}\lnot X^j_ie\big) \rightarrow \big(\bigvee_{J\preceq I}P_{(J,T)} e\big)\Big)\right).    
	\end{align*}
	This formula checks that if an edge $e$ with $T_e = T$ is in the solution assigned to the sets $S_i^{\tr(i,e)}$, for all indices $i$ in some index set $I$, then it is possible to delay the occurrences of edge $e$ in $\gcal$ to time steps in $I$. 
	
	We now define additional auxiliary formula that checks that edge if an edge $e$ is in the set $S^j_i$, then $\tr(i,e)= j$. Let $\mathcal{T}_i^j = \{ (t_1, \ldots, t_D) \mid t_i = j \wedge \forall i'\in D~ 1\le t_{i'}\le D+1 \}$, that is all possible vectors $T$ in a label of an edge such that $t_i = j$. The following formula then check that if $e\in S_i^j$, then $e$ has a label $(I,T)$ such that at the $i$-th position of $T$ is $j$, i.e., $\tr(i,e)=j$, from the construction of $H$.
	\[
	   \checkTraversal(e) = \bigwedge_{i,j\in D}\Big(X_i^j e \rightarrow \bigvee_{I\subseteq [D], T\in \mathcal{T}_j^j}P_{(I,T)}e\Big).
	\]
	
	To check that there is path from a source vertex $s$ to a target vertex $t$ in the delayed graph $\gcal[X_1, \ldots, X_D]$ (where $X_i= X_i^1\cup \cdots \cup X_i^D$),  we will want $X_i$ to contain actually all edges that are used by some source-target path at the time step $i$ even if this edge is not delayed and had label $i$ in the original temporal graph. For $s,t\in V(G)$, $I = \{i_1,\ldots, i_q\}$, and $T= (t_1, \ldots, t_{q})$ such that $i_1 +t_1\le i_2$,  $i_2+ t_2\le i_3$, $\ldots$, $i_{q-1}+t_{q-1}\le i_q$, and $i_{q}+t_{q}\le D$ we define an auxiliary formula 
	\begin{align*}
		\checkPath(s,t,I, T) := \exists e_1, \ldots, e_q, v_1,\ldots v_{q-1} Ise_1\land Ite_q \land
		 \bigwedge_{j\in [q]} X_{i_j}^{t_j} e_j\land 	\bigwedge_{j\in [q-1]} \Big(I v_j e_j\land I v_j e_{j+1}\Big). 
	\end{align*}
	This formula checks that there is a path from $s$ to $t$ such that the $j$-th edge on the path is in the time step $i_j$ and its traversal time is $t_j$.
	Finally, note that we can check that $X_i$ is a set of edges by auxiliary formula 
	\[\isEdgeSet(X) :=\forall e Xe\rightarrow Ee.\]
	Now we are ready to construct the formula $\Phi(X_1^1,\ldots,X_D^D).$
	\begin{align*}
		\Phi(X_1^1,\ldots,X_D^D) & = \bigwedge_{i,j\in [D]} \isEdgeSet(X_i^j) \land \\
		& (\forall e \bigwedge_{I\subseteq [D]}\checkEdge(e,I)) \land \\
		& (\forall e \checkTraversal(e)) \land \\
		& (\forall s,t (P_S s \land V t) \rightarrow \\
		& \quad\quad\quad (\bigvee_{I\subseteq [D], T\in [D]^D} \checkPath(s,t,I,T))). 
	\end{align*}
	
	It is easy to see that we can construct graph $H$ in $\mathcal{O}(D\cdot(|V(H)|+|E(H)|))$ and $\Phi(X_1^1,\ldots,X_D^D)$ is a fixed formula whose construction does not depend on the temporal graph on the input and that can be constructed in time that depends only on $D$.
	\ifshort It only remains to show that \(H\models \Phi(S_1^1,\ldots, S_D^D)\) if and only if the temporal graph \(\gcal[S_1, \ldots, S_D]\), where $S_i = S_i^1\cup\cdots\cup S_i^D$ satisfies for all $v\in S$ that $\reachtime(v, \gcal[S_1, \ldots, S_D])\le D$. \fi 
	\fi
	\iflong	
	By Fact~\ref{fact:MSO} and because $D$ and $\tw(G)$ are constants, we can in linear time either find $S_1^1,\ldots, S_D^D$ such that \(H\models \Phi(S_1^1,\ldots, S_D^D)\) or decide that such sets do not exists. It remains to show that if \(H\models \Phi(S_1^1,\ldots, S_D^D)\) then the temporal graph \(\gcal[S_1, \ldots, S_D]\) satisfies for all $v\in S$ that $\reachtime(v, \gcal[S_1, \ldots, S_D])\le D$.
    On the other hand, we need to show that if there exists a temporal graph \(\gcal'\) obtained from $\gcal$ by delaying, such that for all $v\in S$ it holds that $\reachtime(v, \gcal')\le D$, then there exist sets $S_1,\ldots, S_k$ such that \(H\models \Phi(S_1,\ldots, S_D)\).

	First, let $S_1^1,\ldots, S_D^D$ be such that \(H\models \Phi(S_1,\ldots, S_D)\). For an edge $e$, let $I\subset[D]$ be such that for all $i\in [D]$ it holds that  $e\in S_i$ if and only if $i\in I$. Moreover let $(J,T_e)$ be the label of edge $e$ in $H$.
	From the construction of  \(\Phi(S_1^1,\ldots, S_D^D)\) it follows that \(\checkEdge(e,I)\) is satisfied. Moreover for $T=T_e$, \((\bigwedge_{i\in I}X^{t_i}_ie\land \bigwedge_{i\notin I, j\in [D]}\lnot X_i^je)\) is true by the choice of $I$, hence \((\bigvee_{J'\preceq I}P_{(J',T)} e)\) has to be true. However the edge $e$ has only single label $(J,T)$ in $H$. It follows that $J\preceq I$. In other words, $I = \{i_1,\ldots, i_q\}$ such that $i_1< i_2\cdots< i_q$, $J = \{j_1,\ldots, j_p\}$ such that $j_1< j_2\cdots< j_p$, $q\le p$, and for all $r\in [q]$ we have $j_r\le i_r$. Hence, we can for all $r\in [q]$ delay $j_r$ to $i_r$. It follows that we can delay every edge $e$ such that if $e\in S_i^j$, then  $e\in E'_i$, i.e., $e$ is available at time step $i$ in $\gcal[S_1, \ldots, S_D]$. It follows that $\gcal[S_1, \ldots, S_D]$ can indeed be obtained from $\gcal$ by delaying. Morover, since $\checkTraversal(e)$ is also satisfied, it follows that if $e$ is in $S_i^j$, then indeed $\tr(i,e)=j$, that is the traversal time of $e$ at the time step $i$ is $j$. 
	Now, let $s\in S$ be a source and $t\in V(G)$ a target vertex. We show that we can reach $t$ from $s$ by the time step $D$ in $\gcal[S_1, \ldots, S_D]$.  From the construction of \(\Phi(S_1,\ldots, S_D)\) it follows that \(\bigvee_{I\subseteq [D], T\in [D]^D} \checkPath(s,t,I,T)\) is satisfied. Hence there exists $I\subseteq [D]$ and $T\in [D]^D$ such that $\checkPath(s,t,I,T)$ holds. Let  $I = \{i_1,\ldots, i_q\}$ and $T\in [D]^D$ such that $i_j+t_j\le i_{j+1}$ for $j\in [q-1]$ and $i_q+t+q\le D$. Then there exists edges $e_1, \ldots, e_q$ such that $e_j$ is on $S_{i_j}^{t_j}$ and hence is available at time step $i_j$ in $\gcal[S_1, \ldots, S_D]$ and its traversal time in time step $i_j$ is $t_j$. Moreover, $e_1$ is incident with $s$, $e_q$ is incident with $t$ and two consecutive edges share a vertex. Therefore, $P=(e_j,{i_j})_{j=1}^q$ is a temporal path from $s$ to $t$ and $t$ is reachable from $s$ by time $D$. 
	
	On the other hand, let $\gcal'=(G, \ecal')$ be a temporal graph obtained from delaying some edges in $\gcal$ such that for all   $v\in S$ we have $\reachtime(v, \gcal')\le D$. We show that \(H\models \Phi(E_1^1,\ldots, E_D^D)\), where $E_i^j$ is the set of edges in $E_i'$ with $\tr(i,e)=j$. 
    Clearly, each $E_i^j$ is a set of edges, so \(\isEdgeSet(E_i^j)\) is satisfied for all $i\in [D]$. Now let $e\in H$, $I\subseteq [D]$ be such that $e\in E_i^j$ if and only if $i\in I$ and $\tr(i,e)=j$, and $(J,T_e)\subseteq [D]$ be the label of $e$ in $H$ (note that $e\in E_i$ if and only if $i\in J$). Since, $\gcal'$ is obtained from $\gcal$ by delaying, it is easy 
	to see that $J\preceq I$. Now we show that \((\bigwedge_{I'\subseteq [D]}\checkEdge(e,I'))\). If $T\neq T_e$, then \(\bigvee_{J\subset [D]}P_{(J,T)}e\) is false and implication is satisfied, 
	hence we only need to verify it for $T=T_e$. If $I'\neq I$, then \((\bigwedge_{i\in I}X_i^{t_i}e\land \bigwedge_{i\notin I,j\in [D]}\lnot X_i^je)\) is false and hence \(\checkEdge(e,I')\) is true.  
	Else, $J\preceq I$ and hence \( (\bigvee_{J'\preceq I}P_{J',T} e)\) is true, therefore $\checkEdge(e,I)$ is also true. It follows that  \((\bigwedge_{I'\subseteq [D]}\checkEdge(e,I'))\) is satisfied for $e$.  As we chose $e$ 
	arbitrarily, it follows that \((\forall e \bigwedge_{I\subseteq [D]}\checkEdge(e,I))\) is satisfied. It is easy to see that $\forall e \checkTraversal(e)$ is satisfied as well. Finally, it remains to show that $(\forall s,t (P_S s \land V t) \rightarrow (\bigvee_{I\subseteq [D], T\subseteq [D]^D} \checkPath(s,t,I,T)))$ is satisfied. Let $s\in S$ and $t\in V(G)$ (for any 
	other choice of $s$ and $t$ is the subformula satisfied vacuously) and let $P = (e_j, i_j)_{j=1}^q$ be a temporal path from $s$ to $t$ witnessing that $t$ is reachable from $s$ by time $D$. It is straightforward to verify that $\checkPath(s,t,\bigcup_{j\in [q]}\{i_j\}, tr(i_j,e_j)_{j\in [q]})$ is true, 
    witnessed by edges $e_1,\ldots, e_q$ and vertices $v_j=e_j\cap e_{j+1}$, for $j\in [q-1]$. Hence $(\forall s,t (P_S s \land V t) \rightarrow (\bigvee_{I\subseteq [D], T\in [D]^D} \checkPath(s,t,I,T)))$ is satisfied and \(H\models \Phi(E_1^1,\ldots, E_D^D)\).
\end{proof}
\fi

\subsection{Parallel Paths with two sources}
In this section, we study graphs composed by parallel paths $P_1, P_2, \ldots$ with common endpoints and two sources, where the endpoints are the sources, denoted $A$ and $B$. 
In what follows, for a clearer and more intuitive presentation, we will consider every source to be represented by a distinct color. Thus, the goal is to design algorithms that minimize the time every vertex is reached by all colors.
Our algorithm for this case relies on guessing the time that sources reach each other in an optimal solution. 
For each combination of guesses for the two sources, the algorithm also guesses which paths between $A$ and $B$ achieve the time of the optimal solution. Again, for each guessed path, the algorithm guesses how each source/color optimally reaches every vertex in the rest of the paths.
\iflong
\myalgo{myalgo:Two-Sources-Distinct-Paths}
\label{alg:two-sources-distinct-paths}
\begin{tcolorbox}[title=Algorithm~\ref{myalgo:Two-Sources-Distinct-Paths} \ifshort $(\star)$ \fi (Two Sources Parallel Paths)]

\noindent Let $P_1,P_2,\ldots, P_k$ be the parallel paths of the graph $G$ given as an input.
\noindent For every pair of time steps $r_1, r_2 \in [t_{max}+n-1]^2$ do the following.
        
     For every combination of paths $P_x,P_y$ (note that $P_x$ could be equal to $P_y$) between sources $A$ and $B$:
        \begin{enumerate}
            \item Decide whether $A$ can reach $B$ by time step $r_1$ via path $P_x$ and $B$ can reach $A$ by time step $r_2$ via path $P_y$. 
            \item If the answer is ``Yes'' for paths $P_x,P_y$ then:
            \begin{enumerate}
                \item For every path $P_w$ where $P_w\neq (P_x,P_y)$, take the subgraph defined by $P_w$ and solve the \reachfast problem on each subgraph by running a slight generalization of the algorithm for trees 5 times. The generalization differs to the original algorithm by having the sources become present at particular time steps. For each execution the time at which sources become present in the graph differs as follows:
                
                \iflong
                \begin{enumerate}
                    \item Color $c_A,c_B$ available at $A$ at time step $r_2+1$ and at $B$ at time step $r_1+1$; 
                    \item Color $c_A$ available at $A$ at time $1$ and $B$ at time $r_2+1$ and color $c_B$ available at $B$ at time step 1;
                    \item Color $c_B$ available at $B$ at time $1$ and $A$ at time $r_2+1$ and color $c_A$ available at $A$ at time step 1;
                    \item Color $c_A$ available at $A$ at time $r_2+1$ and $B$ at time $r_1+1$ and color $c_B$ available at $B$ at time step 1;
                    \item Color $c_B$ available at $B$ at time $r_1+1$ and $A$ at time $r_2+1$ and color $c_A$ available at $A$ at time step 1;
                \end{enumerate}
                After these 5 executions, we keep  the solution with minimum reaching-time for subgraph $P_w$.
                
                \fi
            \end{enumerate}
        \end{enumerate}
After the algorithm finishes, we output the solution for pair \((r_1, r_2)\) minimizing the maximum reaching-time over all paths $P_w$.

\end{tcolorbox}

\fi

\iflong
\begin{theorem}

Algorithm \ref{alg:two-sources-distinct-paths} solves the \reachfast problem on graphs that consist of several parallel paths with two common endpoints, where the two endpoints are the sources.
\fi
\ifshort
\begin{theorem}[$\star$]

    There is a polynomial time algorithm for the \reachfast problem on graphs that consist of several parallel paths with two common endpoints, where the two endpoints are the sources.
    \fi
\end{theorem}
\iflong
\begin{proof}
Let us discuss the correctness of the algorithm. Note that whenever color $c_A$ (resp. $c_B$) reaches an edge $w_iw_j$ at time step $t_k$ and the edge has two labels at time steps $t_l,t_m$, where $(t_l,t_m)<t_k$, then it is always optimal for $c_A$ (resp. $c_B$) to delay and use label $t_l$ since label $t_m$ will remain available for color $c_B$ (resp. $c_A$) and there will be no delay on the reaching-time of $c_B$ (resp. $c_A$). Thus, if every edge on a path has two labels, we argue that there exists an optimal solution that includes reaching-time from source $A$ (resp. $B$) to source $B$ (resp. $A$), called $r_1$ (resp. $r_2$), where $1\leq r_1,r_2 \leq (t_{max}+n-1)$, since after time step $t_{max}$ every label is available to be delayed and the maximum distance between the two sources is $n-1$. If there is a path $P_x$ that contains an edge with only one label, then either (i) colors $c_A,c_B$ will traverse that edge together, or (ii) only one of the colors $c_A,c_B$ will use path $P_x$. Either way, there exists an optimal solution, where the reaching-time from source $A$ (resp. $B$) to source $B$ (resp. $A$) will be $1\leq r_1,r_2 \leq (t_{max}+n-1)$ and the algorithm will go through every such reaching-time. 

For each valid reaching-time from source to source, the algorithm then tries to guess the reaching-time for every other vertex on the rest of the paths $P_w$. W.l.o.g. let us assume that in the optimal solution, the reaching-time from source $A$ (resp. $B$) to source $B$ (resp. $A$) is $r_1$ (resp. $r_2$) through path $P_x$ (resp. $P_y$) and $r_1\geq r_2$. For each vertex $w_i$ on each path $P_w=(w_1,w_2,\ldots w_{n-2})$ that source $A$ (resp. $B$) has to reach, there are two possible ways this can happen: through path $(A,w_1,w_2,\ldots,w_i)$ or (ii) through path $(A,P_x,B,w_{n-2},w_{n-3},\ldots w_{i})$. Consider now the last vertex $w_j$ to be reached by both sources $A,B$ on path $P_w$ in an optimal solution. There are $4$ distinct cases showing how $w_j$ was reached: 
\begin{enumerate}
    \item $w_j$ was reached by source $A$ via path $(A,w_1,w_2,\ldots,w_j)$ and by source $B$ via path $(B,P_y$ $,A,w_1,w_2,\ldots,w_j)$;
    \item $w_j$ was reached by source $A$ via path $(A,w_1,w_2,\ldots,w_j)$ and by source $B$ via path $(B,w_{n-2}$ $,w_{n-3},\ldots, w_{i})$;
    \item $w_j$ was reached by source $A$ via path $(A,P_x,B,w_{n-2},w_{n-3},\ldots, w_{j})$ and by source $B$ via path \\ $(B,P_y,A,w_1,w_2,\ldots,w_i)$;
    \item $w_j$ was reached by source $A$ via path $(A,P_x,B,w_{n-2},w_{n-3},\ldots, w_{j})$ and by source $B$ via path \\ $(B,w_{n-2},w_{n-3},\ldots w_{j})$.
\end{enumerate}
 
Consider Case 1 and w.l.o.g. assume that source $A$ reached vertex $w_j$ at time step $t_x$ and source $B$ reached $w_j$ at time step $t_y>t_x$. Note that if we decided to delay the reaching-time of $A$ up to $t_x$, by having color $c_A$ wait at source $A$ for color $c_B$ and then continue together, the reaching-time of every vertex $(w_1,w_2,\ldots,w_g)$ will not increase the value of the overall solution. On the other hand, color $c_A$ may reach vertices $(w_{j+1},\ldots,w_{n-2})$ at a later time step and decrease the value of the overall solution. But if that is true, color $c_A$ should start traversing path $P_w$ without waiting for color $c_B$ at any point since waiting would increase the value of the overall solution as stated before. Therefore there exists an optimal solution where color $c_A$ waits for color $c_B$ at vertex $A$ or there exists an optimal solution where $c_A$ never waits for color $c_B$.
Cases 2,3,4 can be argued in a similar fashion where in every case there exists an optimal solution where either color $c_A$ (resp. $c_B$) waits for color $c_B$ (resp. $c_A$) at source $A$ (resp. $B$), or color $c_A$ (resp. $c_B$) never waits for color $c_B$ (resp. $c_A$). Algorithm \ref{alg:two-sources-distinct-paths} checks every such case and specifically: in Step (i) both colors $c_A,c_B$ wait for the other color. In Steps (ii) and (iii), no color waits for the other. In Step (iv) color $c_A$ waits for color $c_B$ and in Step (v) color $c_B$ waits for color $c_A$. Thus, the algorithm always finds an optimal solution.

Let us now discuss the complexity of the algorithm. By description, going through every possible reaching-time from source $A$ (resp. $B$) to source $B$ (resp. $A$) requires $O(t_{max}+n-1)^2$ repetitions. For every such reaching-time, finding the correct paths requires $O(n^2)$ repetitions, since there are $n^2$ combinations of paths $P_x,P_y$. For every such combination, the algorithm finds the optimal reaching-time for every path $P_w$ by checking each path $5$ times. Since there are $n$ paths at most, this step requires $O(n)$ repetitions. In total, Algorithm \ref{alg:two-sources-distinct-paths} runs in polynomial time.
\fi
\ifshort
\begin{proof}[Proof sketch]
The algorithm goes through every possible pair of reaching-times \((r_1, r_2)\) that sources \(A\) and \(B\) can have to reach one another. 
    This process requires $O(t_{max}+n-1)^2$ repetitions.
    Note that this process runs through every possible solution of the \reachfast problem since by time step $t_{max}$ every label on the graph is available to be delayed and the static distance between the sources cannot be more than $n-1$. Since every label is available, if every edge has two labels and there are two sources on the graph, there is no reason for a color to delay using an edge, since there is always a second label available for the other color. If there is only one label on an edge of a path $P_x$, there are only two valid choices: either (i) both colors traverse the one label edge together (possibly in the opposite direction), or (ii) only one color uses path $P_x$. Thus the algorithm runs through every possible scenario that might yield a solution to the \reachfast problem. 
    
    In Step $1$, the algorithm checks every possible combination of paths that can achieve the reaching-time imposed by each for-loop. This can be achieved in polynomial time, since every loop checks a combination of two paths, and any graph can have at most $n$ parallel paths between the two sources. 
    
    In Step $2$, for every path combination of $P_x,P_y$ that can achieve the desired reaching-time, the algorithm uses the multiple source tree algorithm as a subroutine to transmit the color to the rest of the edges on the paths. Specifically, for each path \(P_w\) we need to distinguish between several cases depending on whether the sources $A$ and $B$ reach all the vertices of \(P_w\) directly on the path \(P_w\) or, for example, the source \(A\) uses \(P_x\) to reach \(B\) before reaching some vertices of the path \(P_w\). Unfortunately in the latter case we cannot use Theorem~\ref{thm:trees} directly, as the vertex \(B\) becomes available for the "color" of source \(A\) only at time step \(r_1\), but a minor modification makes the algorithm to work.
\fi
\end{proof}

\section{Conclusion}
\label{sec:conclusions}
We view our paper as part of a greater agenda whose goal is to make existing infrastructures more efficient, by making minimal changes to the currently-adopted solutions. 
We introduced and studied the complexity  of the \reachfast problem, where the goal was to minimize the maximum reaching-time of a set of sources towards every vertex of the graph using shifting operations. %

Since we study optimization problems that as we proved are \NP-hard, instead of focusing on restricted classes of instances, one can ask to relax the optimality condition and ask for some approximate solutions. Here we would like to point out that Theorem~\ref{thm:two-sources-hard} implies the problem is \APX-hard even for two sources, however existence of a constant approximation for some constant $c>7/6$ is an interesting open question. In both, \reachfast(k) and \reachfasttotal(k) there are two criteria one might wish to optimize: the maximum reaching-time, i.e., a deadline \(D\), and \(k\). Our reduction from \textsc{Hitting Set} is approximation-preserving for \(k\) and hence a \((1-o(1))\log (|V(G)|)\)-approximation algorithm for value \(k\) is unlikely, unless $\PP = \NP$. Similarly, for fixed \(k\), the reduced instance we obtain either admits a solution with \(D=3\) or does not admit a solution, and polynomial time algorithm that approximates \(D\) for fixed \(k\) is again highly unlikely.

Our work also creates several interesting directions for future research that can be studied under the shifting operations. Instead of minimizing the maximum reaching-time of any source, we could try to minimize the average time a source needs to reach all the vertices of the graph. Another objective is to minimize the maximum, or average, time a vertex is reached by {\em any} of the sources; this would be desirable, for example, in the case where a company has several warehouses over a country. The dual version of the problem is intriguing as well. Here, instead of sources, we have a set of {\em sinks}, i.e., every vertex wants to reach them as fast as possible. This model can capture for example the scenario where the sinks correspond to hospitals. 

At a different dimension, someone could focus on other restricted classes of temporal graphs that could be proven tractable. We believe that the class implicitly introduced in~\cite{teec-paper} is one such class; there the temporal graph is given as a collection of ``few'' temporal paths/walks/trails.
While some of our results should work for some cases of this setting, we expect that novel algorithmic techniques are required in order to tackle them.

\bibliographystyle{named}
\bibliography{references}

\begin{thebibliography}{}

\bibitem[\protect\citeauthoryear{Arnborg \bgroup \em et al.\egroup
  }{1991}]{ArnborgLagergrenSeese91}
Stefan Arnborg, Jens Lagergren, and Detlef Seese.
\newblock Easy problems for tree-decomposable graphs.
\newblock {\em J. Algorithms}, 12(2):308--340, 1991.

\bibitem[\protect\citeauthoryear{Bodlaender}{1996}]{Bodlaender96}
Hans~L. Bodlaender.
\newblock A linear-time algorithm for finding tree-decompositions of small
  treewidth.
\newblock {\em {SIAM} J. Comput.}, 25(6):1305--1317, 1996.

\bibitem[\protect\citeauthoryear{Braunstein and
  Ingrosso}{2016}]{braunstein2016}
Alfredo Braunstein and Alessandro Ingrosso.
\newblock Inference of causality in epidemics on temporal contact networks.
\newblock {\em Scientific reports}, 6:27538, 2016.

\bibitem[\protect\citeauthoryear{Casteigts \bgroup \em et al.\egroup
  }{2021}]{casteigts2021finding}
Arnaud Casteigts, Anne-Sophie Himmel, Hendrik Molter, and Philipp Zschoche.
\newblock Finding temporal paths under waiting time constraints.
\newblock {\em Algorithmica}, pages 1--49, 2021.

\bibitem[\protect\citeauthoryear{Courcelle}{1990}]{Courcelle90}
Bruno Courcelle.
\newblock The monadic second-order logic of graphs. {I}. recognizable sets of
  finite graphs.
\newblock {\em Inf. Comput.}, 85(1):12--75, 1990.

\bibitem[\protect\citeauthoryear{Cygan \bgroup \em et al.\egroup
  }{2015}]{CyganFKLMPPS15}
Marek Cygan, Fedor~V. Fomin, Lukasz Kowalik, Daniel Lokshtanov, D{\'{a}}niel
  Marx, Marcin Pilipczuk, Michal Pilipczuk, and Saket Saurabh.
\newblock {\em Parameterized Algorithms}.
\newblock Springer, 2015.

\bibitem[\protect\citeauthoryear{Darmann and D{\"o}cker}{2020}]{DD20}
Andreas Darmann and Janosch D{\"o}cker.
\newblock On a simple hard variant of {N}ot-{A}ll-{E}qual 3-{S}at.
\newblock {\em Theoretical Computer Science}, 815:147--152, 2020.

\bibitem[\protect\citeauthoryear{Deligkas and Potapov}{2022}]{DP20}
Argyrios Deligkas and Igor Potapov.
\newblock Optimizing reachability sets in temporal graphs by delaying.
\newblock {\em Information and Computation}, 285:104890, 2022.

\bibitem[\protect\citeauthoryear{Deligkas \bgroup \em et al.\egroup
  }{2023}]{DBLP:conf/ijcai/DeligkasES23}
Argyrios Deligkas, Eduard Eiben, and George Skretas.
\newblock Minimizing reachability times on temporal graphs via shifting labels.
\newblock In {\em Proceedings of the Thirty-Second International Joint
  Conference on Artificial Intelligence, {IJCAI} 2023, 19th-25th August 2023,
  Macao, SAR, China}, pages 5333--5340. ijcai.org, 2023.

\bibitem[\protect\citeauthoryear{Deligkas \bgroup \em et al.\egroup
  }{2024}]{teec-paper}
Argyrios Deligkas, Michelle D{\"{o}}ring, Eduard Eiben, Tiger{-}Lily Goldsmith,
  George Skretas, and Georg Tennigkeit.
\newblock How many lines to paint the city: Exact edge-cover in temporal
  graphs.
\newblock {\em CoRR}, abs/2408.17107, 2024.

\bibitem[\protect\citeauthoryear{Downey and Fellows}{2013}]{DowneyFellows13}
Rodney~G. Downey and Michael~R. Fellows.
\newblock {\em Fundamentals of parameterized complexity}.
\newblock Texts in Computer Science. Springer, 2013.

\bibitem[\protect\citeauthoryear{Enright and Kao}{2018}]{ENRIGHT201888}
Jessica Enright and Rowland~Raymond Kao.
\newblock Epidemics on dynamic networks.
\newblock {\em Epidemics}, 24:88 -- 97, 2018.

\bibitem[\protect\citeauthoryear{Enright and Meeks}{2018}]{COA2015}
Jessica Enright and Kitty Meeks.
\newblock Deleting edges to restrict the size of an epidemic: a new application
  for treewidth.
\newblock {\em Algorithmica}, 80(6):1857--1889, 2018.

\bibitem[\protect\citeauthoryear{Enright \bgroup \em et al.\egroup
  }{2021a}]{Deleting_Edges}
Jessica~A. Enright, Kitty Meeks, George~B. Mertzios, and Viktor Zamaraev.
\newblock Deleting edges to restrict the size of an epidemic in temporal
  networks.
\newblock {\em Journal of Computer and System Sciences}, 119:60--77, 2021.

\bibitem[\protect\citeauthoryear{Enright \bgroup \em et al.\egroup
  }{2021b}]{reordering}
Jessica~A. Enright, Kitty Meeks, and Fiona Skerman.
\newblock Assigning times to minimise reachability in temporal graphs.
\newblock {\em J. Comput. Syst. Sci.}, 115:169--186, 2021.

\bibitem[\protect\citeauthoryear{Kempe \bgroup \em et al.\egroup
  }{2002}]{kempe02}
David Kempe, Jon Kleinberg, and Amit Kumar.
\newblock Connectivity and inference problems for temporal networks.
\newblock {\em Journal of Computer and System Sciences}, 64(4):820--842, 2002.

\bibitem[\protect\citeauthoryear{Klobas \bgroup \em et al.\egroup
  }{2024}]{KMMS22}
Nina Klobas, George~B Mertzios, Hendrik Molter, and Paul~G Spirakis.
\newblock The complexity of computing optimum labelings for temporal
  connectivity.
\newblock {\em Journal of Computer and System Sciences}, 146:103564, 2024.

\bibitem[\protect\citeauthoryear{Li \bgroup \em et al.\egroup
  }{2018}]{li2018go}
Lei Li, Kai Zheng, Sibo Wang, Wen Hua, and Xiaofang Zhou.
\newblock Go slow to go fast: minimal on-road time route scheduling with
  parking facilities using historical trajectory.
\newblock {\em The VLDB Journal}, 27(3):321--345, 2018.

\bibitem[\protect\citeauthoryear{Molter \bgroup \em et al.\egroup
  }{2024}]{MRZ21}
Hendrik Molter, Malte Renken, and Philipp Zschoche.
\newblock Temporal reachability minimization: Delaying vs. deleting.
\newblock {\em Journal of Computer and System Sciences}, 144:103549, 2024.

\bibitem[\protect\citeauthoryear{Thejaswi and
  Gionis}{2020}]{thejaswi2020restless}
Suhas Thejaswi and Aristides Gionis.
\newblock Restless reachability in temporal graphs.
\newblock {\em arXiv preprint arXiv:2010.08423}, 2020.

\bibitem[\protect\citeauthoryear{Wang \bgroup \em et al.\egroup }{2015}]{WL+15}
Sibo Wang, Wenqing Lin, Yi~Yang, Xiaokui Xiao, and Shuigeng Zhou.
\newblock Efficient route planning on public transportation networks: A
  labelling approach.
\newblock In {\em Proc. of {SIGMOD}}, pages 967--982. ACM, 2015.

\bibitem[\protect\citeauthoryear{Whitbeck \bgroup \em et al.\egroup
  }{2012}]{Whitbeck2012}
John Whitbeck, Marcelo Dias~de Amorim, Vania Conan, and Jean-Loup Guillaume.
\newblock Temporal reachability graphs.
\newblock In {\em Mobicom}, pages 377--388, 2012.

\bibitem[\protect\citeauthoryear{Wu \bgroup \em et al.\egroup
  }{2012}]{shortest-experimental}
Lingkun Wu, Xiaokui Xiao, Dingxiong Deng, Gao Cong, Andy~Diwen Zhu, and
  Shuigeng Zhou.
\newblock Shortest path and distance queries on road networks: An experimental
  evaluation.
\newblock {\em Proc. VLDB Endow.}, 5(5):406–417, 2012.

\bibitem[\protect\citeauthoryear{Wu \bgroup \em et al.\egroup }{2014}]{WC+14}
Huanhuan Wu, James Cheng, Silu Huang, Yiping Ke, Yi~Lu, and Yanyan Xu.
\newblock Path problems in temporal graphs.
\newblock {\em Proc.\ of {VLDB} Endowment}, 7(9):721--732, 2014.

\bibitem[\protect\citeauthoryear{Wu \bgroup \em et al.\egroup }{2016}]{WH+16}
Huanhuan Wu, Yuzhen Huang, James Cheng, Jinfeng Li, and Yiping Ke.
\newblock Reachability and time-based path queries in temporal graphs.
\newblock In {\em Proc.\ of {ICDE}}, pages 145--156. IEEE, 2016.

\end{thebibliography}

%
\end{document}